\documentclass[12pt]{article}%
\pdfoutput=1
\newcommand{\sv}[1]{}%
\usepackage{etoolbox}

\usepackage[table]{xcolor}  %
\usepackage[T1]{fontenc}
\usepackage[utf8]{inputenc}
\usepackage[colorlinks=true,allcolors=blue]{hyperref}
\usepackage{graphicx} %
\usepackage{subcaption}
\usepackage{lmodern}
\usepackage{fullpage}
\usepackage{authblk}
\usepackage{amsmath}
\usepackage{amsthm}
\usepackage{xspace}
\usepackage{adjustbox}
\usepackage{mathtools}
\usepackage{makecell}
\usepackage{paralist} %
\usepackage{mwe}\usepackage[hypcap=false]{caption} %
\usepackage[capitalize]{cleveref}

\newcommand{\NP}[0]{\ensuremath{\mathrm{NP}}\xspace}

\usepackage[color=Olive]{todonotes}

\newcommand{\keywords}[1]{\smallskip\begin{center}\textit{Keywords: #1}\end{center}\bigskip}

\newtheorem{lemma}{Lemma}[section]
\newtheorem{theorem}[lemma]{Theorem}
                              
\newtheorem{proposition}[lemma]{Proposition}

\newcommand{\unkn}[0]{\cellcolor{blue!25}}
\newcommand{\ea}[0]{\cellcolor{green!25}}
\newcommand{\ha}[0]{\cellcolor{red!30}}
\newcommand{\current}[0]{\cellcolor{orange!30}}

\newcommand{\oldM}[0]{M^*}
\newcommand{\Gcore}[0]{\mathcal{\bar{G}}}

\usepackage[absolute]{textpos}
\title{%
  Constricting
 the Computational Complexity Gap \\of the \textsc{$4$-Coloring} Problem in $(P_t,C_3)$-free Graphs\thanks{
    T.~Masařík was supported by the Polish National Science Centre SONATA-17 grant number 2021/43/D/ST6/03312.
    J.~Masaříková was supported by the Polish National Science Centre Preludium research project no.\ 2022/45/N/ST6/04232.
    }
}

\author[1]{Justyna Jaworska}
\author[1]{Bartłomiej Kielak\thanks{Current affiliation: Leipzig University, Leipzig, Germany.}}
\author[3]{\\Tomáš Masařík}
\author[3]{Jana Masaříková}

\affil[1]{Jagiellonian University, Krakow, Poland}
\affil[ ]{\texttt{justynajoanna.jaworska@student.uj.edu.pl}}
\affil[ ]{\texttt{bkielak@alumni.uj.edu.pl}}
\affil[2]{University of Warsaw, Poland}%
\affil[ ]{\texttt{\{masarik,jnovotna\}@mimuw.edu.pl}}
\date{}

\begin{document}

\maketitle

\begin{abstract}
The $k$-\textsc{Coloring} problem on hereditary graph classes has been a deeply researched problem over the last decade.
A hereditary graph class is characterized by a (possibly infinite) list of minimal forbidden induced subgraphs.
We say that a graph is \emph{$(H_1,H_2,\ldots)$-free} if it does not contain any of $H_1,H_2,\ldots$ as induced subgraphs.
The complexity landscape of the problem remains unclear even when restricting to the case $k=4$ and classes defined by a~few forbidden induced subgraphs.
While the case of only one forbidden induced subgraph has been completely resolved lately, the complexity when considering two forbidden induced subgraphs still has a couple of unknown cases. 
In particular, \textsc{$4$-Coloring} on $(P_6,C_3)$-free graphs is polynomial while it is \NP-hard on $(P_{22},C_3)$-free graphs.

We provide a reduction showing \NP-completeness of \textsc{$4$-Coloring} on $(P_t,C_3)$-free graphs for $19\leq t\leq 21$, thus %
constricting the gap of cases whose complexity remains unknown.
Our proof includes a computer search ensuring that the graph family obtained through the reduction is indeed $P_{19}$-free.
\end{abstract}

\section{Introduction}

Graph coloring is a well-established concept in both Graph Theory and Theoretical Computer Science. 
A \emph{$k$-coloring} of a graph $G=(V,E)$ is defined as a mapping $c:V\to \{1,\ldots,k\}$ which is \emph{proper}, i.e., it assigns distinct colors to vertices $u,v\in V$ if $uv\in E$. 

The \textsc{$k$-Coloring} problem asks whether a given graph $G$ admits a $k$-coloring. We also define the \textsc{$k$-Precoloring Extension} problem, in which, apart from $G$, a subset $W \subseteq V(G)$ with a \emph{precoloring} $c_W : W \rightarrow \{1,\dots,k\}$ is given, and the task is to determine whether there exsists a $k$-coloring of $G$ which agrees with $c_W$ on $W$. Lastly, an instance of the \textsc{List $k$-coloring} problem consists of a graph $G$ and a \emph{list assignment} $L$, which assigns to each $v\in V(G)$ a list of admissible colors $L(v) \subseteq \{1,\dots, k\}$.  In that case, the coloring function $c$, in addition to being proper, has to respect the lists, that is, $c(v)\in L(v)$ for every vertex~$v$.

Observe that \textsc{$k$-Coloring} can be viewed as a special case of \textsc{$k$-Precoloring
Extension} which is, in turn, a special case of \textsc{List $k$-Coloring}.
For any $k\ge 3$, \textsc{$k$-Coloring} and, consequently, both \textsc{$k$-Precoloring Extension} and \textsc{list $k$-coloring} are well-known to be \NP-complete~\cite{Karp72}.

As these coloring problems are hard, we turn our attention to graph classes for which efficient algorithms might exist.
A graph class is \emph{hereditary} if it is closed under vertex deletion.
It follows that a graph class  $\mathcal{G}$ is hereditary if and only if $\mathcal{G}$  can be characterized by a unique (though not necessarily finite) set $\mathcal H_{\mathcal G}$ of minimal forbidden induced subgraphs.
When $\{H\}=\mathcal H_{\mathcal G}$, or $\{H_1,H_2,\ldots\}=\mathcal H_{\mathcal G}$, we say that $G\in\mathcal G$ is \emph{$H$-free}, or \emph{$(H_1,H_2,\ldots)$-free}, respectively.
Of particular interest are hereditary graph classes where $\mathcal H_{\mathcal G}$ contains only one or only a very few elements.
For two graphs $H_1$ and $H_2$, we let $H_1+H_2$ denote their disjoint union, and we write $kH$ for the disjoint union of $k$ copies of a graph~$H$.
We let $P_t$ denote the path on $t$ vertices, and $C_\ell$ the cycle on $\ell$ vertices.

The \textsc{$k$-Coloring} problem, along with the \textsc{maximum independent set} problem, has played a pivotal role in algorithm design for specific hereditary graph classes.
A fast and intensive study of this area has led to the development of new tools and the understanding of the complex structure of certain graph classes in recent years.
In particular, great effort was invested into the analysis of classes with only one or a few forbidden graphs.

Focusing on the \textsc{$k$-Coloring} problem with only \emph{one} forbidden induced subgraph, classical results imply that for every $k\ge 3$, \textsc{$k$-Coloring} of $H$-free graphs is \NP-complete if $H$ contains a cycle~\cite{EHK98} or an induced claw~\cite{Holyer81,LevenGalil83}. The only graphs that are not covered by these results are \emph{linear forests}, i.e. disjoint unions of paths. We shall discuss these, distinguishing the cases $k=3$ and $k > 3$.

It would appear that $k=3$ is the easiest to reason about.
However, the complexity of \textsc{$3$-Coloring} is fully known only for linear forests on at most $7$ vertices~\cite{BCMSZ17,KMMNPS20} and some exceptional larger graphs, such as $H=P_6+rP_3$~\cite{CHSZ20}.
The smallest unsolved cases are $H=P_8$ and $H=2P_4$; see~\cite{2P4WG} for a summary of the recent developments.
The discussion above indicates that the \textsc{$3$-Coloring} problem is challenging with the current set of tools even when only a single graph is forbidden, despite a recent breakthrough showing a quasi-polynomial time algorithm for any fixed $t$~\cite{GL20,PPR20}.

In contrast to case $k=3$, the complexity classification for $k>3$ has been resolved almost completely. %
Most notably, the full classification is known for $P_t$-free graphs, as can be seen from the list of the following results.
The \textsc{$4$-Coloring} problem is \NP-complete for $H=P_7$~\cite{Huang16} while for $H=P_6$, the \textsc{$4$-Precoloring Extension} problem is polynomial~\cite{SCZp1,SCZp2} and the \textsc{List $4$-Coloring} problem is \NP-complete~\cite{GPS14_IC}.
For all larger fixed $k>4$, the classification is now complete for the list version: the recent dichotomy~\cite{List-k-coloring} shows that the \textsc{List $k$-Coloring} problem on $H$-free graphs is polynomial-time solvable if and only if $H$ is contained in $rP_3$ for some $r\ge 1$ or in $P_5+rP_1$ for some $r\ge1$; otherwise, it is \NP-hard.

Naturally, the situation is much more complex when a pair of forbidden subgraphs is considered.
There has been a great effort in the characterization of various interesting properties of such graphs; consult the detailed survey by Golovach, Johnson, Paulusma, and Song~\cite{GJPS16}.
A systematic approach to finding the classes where many problems, including coloring, are polynomial, has been made by classification of what classes have bounded clique-width~\cite{cw-survey}, or clique-width of atoms~\cite{atoms}, or even mim-width~\cite{mimw}.
All the above properties guarantee a polynomial-time algorithm for the \textsc{$k$-Coloring} problem (even with unbounded $k$ in the first two cases).
In what follows, we focus on the \textsc{$4$-Coloring} problem on classes with exactly two forbidden induced subgraphs.

\begin{table}[t]
\centering
\begin{tabular}{c|c|c|c|c|c}
& {\small $\ell=3$} 
& {\small \makecell{$\ell=4$\\\cite{GPS14_DAM}}} 
& {\small \makecell{$\ell\in\{5,6\}$\\\cite{HellHuang17}}} 
& {\small \makecell{$\ell=7$\\\cite{HellHuang17}}}
& {\small \makecell{$\ell\ge 8$\\\cite{HellHuang17}}} \\
\hline
{\small $t\le 6$\cite{SCZp1,SCZp2}}      
& {\ea P} & {\ea P} & {\ea P} & {\ea P} & {\ea P} \\
{\small $7\le t \le 8$}         
& \unkn ? & {\ea P} & {\ha \NP{}-c} & \unkn ? & {\ha \NP{}-c} \\
{\small $9\le t\le 18$}
& \unkn ? & {\ea P} & {\ha \NP{}-c} & {\ha \NP{}-c} & {\ha \NP{}-c} \\
{\small $19 \le t\le 21$}     
& {\current \NP{}-c} & {\ea P} & {\ha \NP{}-c} & {\ha \NP{}-c} & {\ha \NP{}-c} \\
{\small $t\ge 22$}     
& {\ha \NP{}-c~\cite{HJP15}} & {\ea P} & {\ha \NP{}-c} & {\ha \NP{}-c} & {\ha \NP{}-c} \\
\end{tabular}
\caption{Summary of complexity results for {\sc $4$-Coloring} (or even \textsc{$4$-Precoloring Extension}) on $(P_t,C_\ell)$-free graphs. The result of this paper is highlighted with an orange background, other hardness results are highlighted in red, polynomial results are green, and unknown is blue.}
\label{tab:ptcl}
\end{table}

\paragraph{Results Overview for $(P_t,C_\ell)$-free Graphs.}
There was a particular attention on the classification of the complexity of the \textsc{$4$-Coloring} problem on $(P_t,C_\ell)$-free graphs.
We summarize the known results in \cref{tab:ptcl}.
Somewhat surprisingly, when $\ell=4$ and $t$ is arbitrary but fixed, the \textsc{$4$-Coloring} (in fact, any \textsc{$k$-Coloring}) is polynomial~\cite{GPS14_DAM}.
When $\ell \geq 5$, the classification is almost fully known aside from two unresolved cases ($\ell=7, t\in\{7,8\}$), \textsc{$4$-Coloring} is \NP-complete whenever $t \geq 7$ \cite{HellHuang17}. For $\ell = 3$, though, the situation is more involved. 

It is known that \textsc{$4$-Precoloring Extension} is polynomial-time solvable in the class of $(P_6, C_3)$-free graphs \cite{SCZp1, SCZp2}. In $2014$, Huang, Johnson and Paulusma \cite{HJP15} proved that \textsc{$4$-Coloring} $(P_{22}, C_3)$-free graphs is \NP-complete, thus improving the previous bound $t\ge 164$ \cite{GPS11}.

All the results above translate to the \textsc{$4$-Precoloring Extension} problem.
Notably, a bit less is known about the \textsc{List-$4$-Coloring} problem.
The polynomial results for \textsc{$4$-Coloring} do not translate, and, for example, for ($P_6,C_\ell$)-free graphs, the list version was shown to be \NP-complete~\cite{HJP15}.

Results in \cref{tab:ptcl} leave open only the following cases: $(P_7,C_7)$-, $(P_8,C_7)$-, and $(P_t,C_3)$-free graphs, for $7\le t\le 21$.
In this paper, 
we settle three open cases with $\ell=3$. 

\begin{theorem}\label{thm:main}
  The \textsc{$4$-Coloring} problem is \NP-complete in the class of $(P_t,C_3)$-free graphs for ${t\ge 19}$.
\end{theorem}

\subsection{Preliminaries}\label{sec:prelim}
\paragraph{Mycielski Construction.}
A classical way to increase the chromatic number of a graph without creating larger cliques is the following \emph{Mycielski construction}~\cite{Mycielski} dating back to 1955. 
Given a graph $G=(V,E)$ with $V=\{v_{0},\dots ,v_{n-1}\}$, the \emph{Mycielski graph} $\mu(G)$ is obtained from $G$ in three steps:  
(i) add a \emph{shadow vertex} $v_{n+i}$ for every $v_{i}\in V$;  
(ii) for every edge $v_{i}v_{j}\in E$, insert the edges $v_{i}v_{n+j}$ and $v_{n+i}v_{j}$;  
(iii) add a new \emph{universal vertex} $v_{2n}$ and join it to all shadow vertices $v_{n},v_{n+1},\dots ,v_{2n-1}$.
The resulting graph has $2n+1$ vertices, contains no new triangles, and satisfies that the chromatic number of $\mu(G)$ is one larger than that of $G$ while the clique number stays the same.
Iterating the operation starting with $M_2\coloneqq K_2$ yields the family $M_{k}\coloneqq \mu^{\,k-2}(K_2)$ ($k\ge 2$), providing for each $k$ a triangle-free graph with chromatic number $k$. The graphs $M_k$ are moreover \emph{critical} -- any proper subgraph of $M_k$ is $(k-1)$-colorable.

\paragraph{Monotone Not-All-Equal-3-SAT problem.}
Given a finite set $X=\{x_{0},\dots ,x_{n}\}$ of Boolean variables and a family $\mathcal{S}\subseteq\binom{X}{3}$ of three-element clauses containing \emph{only positive literals}, the \emph{\textsc{MNAE}-$3$-\textsc{SAT}} problem asks whether there exists an assignment ${\sigma:X\to\{\top,\bot\}}$ such that, for every clause $\{x_i,x_j,x_k\}=S\in\mathcal{S}$, the triple $\{\sigma(x_i),\sigma(x_j),\sigma(x_k)\}$ is \emph{not-all-equal}. 
Equivalently, each clause contains at least one true and at least one false variable.
This problem is long-known to be \NP-complete~\cite{Schaefer78}.

\bigskip
Our reduction is a refinement of one presented in \cite[Theorem 7]{HJP15}, which we first recall.

\subsection{The Original Construction in \cite[Theorem 7]{HJP15}}\label{sec:oldConstruction}

Below, we outline the reduction from \cite{HJP15} which proves \NP-completeness of \textsc{4-Coloring} for $(P_{22}, C_3)$-free graphs. The key ideas are as follows:
\begin{itemize}
    \item Build an instance $J_\phi'$ of \textsc{List $4$-Coloring} corresponding to an input \textsc{MNAE-$3$-SAT} instance $\phi$;
    \item Convert it into an equivalent \textsc{$4$-Coloring} instance $J_\phi^*$ by enforcing color lists with the help of a special \emph{color synchronization gadget}.
\end{itemize}

Let $\phi=(X,\mathcal{S})$ be an instance of \textsc{MNAE}‑3-\textsc{SAT}, where
$X=\{x_{0},\dots,x_{n}\}$ and $\mathcal{S}=\{S_{0},\dots,S_{m}\}$ with $S_{j}=\{x_{j_{0}},x_{j_{1}},x_{j_{2}}\}$.
First, we construct the graph $J_\phi$.
For every clause $S_{j}$, the graph $J_{\phi}$ contains two components $C_{j}$ and $C'_{j}$, each isomorphic to $P_{5}$.  
Writing the vertices of $C_{j}$ in order along the path gives  
\[
a_{j,0}\!-\!b_{j,0}\!-\!a_{j,1}\!-\!b_{j,1}\!-\!a_{j,2},
\]
and similarly for $C'_{j}$,
\[
a'_{j,0}\!-\!b'_{j,0}\!-\!a'_{j,1}\!-\!b'_{j,1}\!-\!a'_{j,2}.
\]
The list assignment $L$ is:
\[
\begin{aligned}
L(a_{j,0}) &= \{1,3\}, & L(b_{j,0}) &= \{2,3\}, &
L(a_{j,1}) &= \{1,2,3\}, & L(b_{j,1}) &= \{2,3\}, &
L(a_{j,2}) &= \{1,2\},\\
L(a'_{j,0}) &= \{0,3\},& L(b'_{j,0}) &= \{2,3\},&
L(a'_{j,1}) &= \{0,2,3\},& L(b'_{j,1}) &= \{2,3\},&
L(a'_{j,2}) &= \{0,2\}.
\end{aligned}
\]

For every variable $x_{i}\in X$, $J_{\phi}$ contains a single \emph{$x$‑type} vertex $x_{i}$ with list $L(x_{i})=\{0,1\}$.
Adjacencies between clauses and variables are defined as follows.
\begin{itemize}%
  \item[(i)] %
        For each $h\in\{0,1,2\}$, add edges $a_{j,h}x_{j_{h}}$ and $a'_{j,h}x_{j_{h}}$.
  \item[(ii)] Moreover, every $x$‑type vertex $x_{i}$ is adjacent to every $b$‑type vertex, i.e., to all vertices in $\{b_{j,0},b_{j,1},b'_{j,0},b'_{j,1}\mid 0\le j\le m\}$.
\end{itemize}

Furthermore, every edge joining an $a$-type with an $x$-type vertex is subdivided; declare each subdivision vertex to be of \emph{$c$-type} and give it the list $\{0,1\}$.  
Denote the resulting list‑colored graph by $J'_{\phi}$ and its list by $L'$. 
It is straightforward to verify that $J'_\phi$ is a positive instance of \textsc{List $4$-Coloring} if and only if $\phi$ is a positive instance of \textsc{MNAE-$3$-SAT}.

Next, we describe further adjustments that eliminate the lists in the construction.
For every vertex $u\in V(J'_{\phi})$ and every color $\gamma\notin L'(u)$, we attach a new pendant vertex $w_{u,\gamma}$ adjacent only to $u$ and set the list of $w_{u,\gamma}$ to $\{\gamma\}$.
This allows us to discard the lists of other vertices and thus, we formed an instance of \textsc{$4$-Precoloring Extension}.
Let $W^{4}$ be the set of all such pendant vertices. %

Let $M_{5}$ be the fifth Mycielski graph, which will play the role of a universal \emph{color synchroniza\-tion gadget}; see \cref{sec:prelim} for basics about this well-known graph construction. 
As $M_5$ is not $4$-colorable, one specific edge is removed.
We denote the modified graph $\oldM$.
$\oldM$ is then $4$-colorable.
Moreover, there exist four mutually non-adjacent vertices $t_{0}, t_{1},t_{2},t_{3} \in V(\oldM)$ such that in each 4-coloring they always receive distinct colors.
Observe that without the triangle-freeness requirement, we could use $K_4$ in place of $\oldM$.

Finally, we obtain $J^{*}_{\phi}$ from $J'_{\phi}$ by the following operations:
\begin{enumerate}
  \item delete every pendant vertex adjacent to a vertex in $B\cup C$;
  \item add a disjoint copy of $\oldM$ and keep the distinguished vertices $t_{0},t_{1},t_{2},t_{3}$;
  \item for every remaining pendant $v\in W^{4}$ with prescribed color $i\in\{0,1,2,3\}$, join $v$ to all $t_{j}$ with $j\neq i$;
  \item connect every vertex of $B$ to $t_{0}$ and $t_{1}$;
  \item connect every vertex of $C$ to $t_{2}$ and $t_{3}$.
\end{enumerate}

\begin{figure}
\centering
    \includegraphics[width=0.9\textwidth]{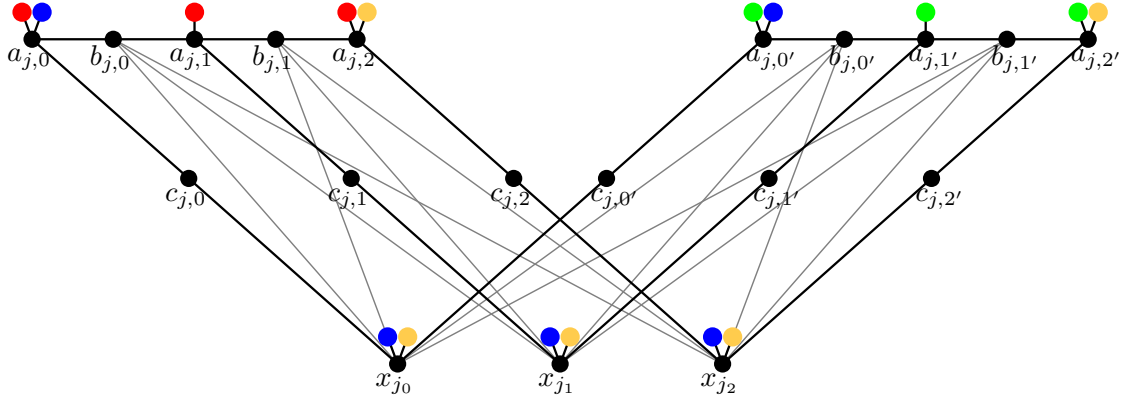}
\caption{A gadget corresponding to a clause $S_j = (x_{j_0}, x_{j_1}, x_{j_2})$ in the construction $J^{*}_{\phi}$ from \cite{HJP15}. Each colored vertex is adjacent to an appropriate subset of $\{t_0, t_1, t_2, t_3\}$, and so are all $b$-type and $c$-type vertices.}
\label{pic:old_gadget_2015}
\end{figure}

The final graph $J^{*}_{\phi}$ is $(C_{3},P_{22})$-free. Moreover, $\phi$ is satisfiable if and only if $J^{*}_{\phi}$ is $4$‑colorable. Figure \ref{pic:old_gadget_2015} shows a gadget that corresponds to a single clause.

Let us finally remark that \cite{HJP15} gives a construction of $P_{21}$ using only 2 vertices of $\oldM$, which disincentivizes improvements of the graph $\oldM$ in their construction.
In this paper, we improve the construction above; in particular, we utilize the color synchronization gadget more.

\section{Proof of Theorem~\ref{thm:main}}

We modify the construction from \cref{sec:oldConstruction}.

\subsection{New Construction Description}\label{sec:newConstruction}

We start with an instance of \textsc{MNAE}‑3-\textsc{SAT} $\phi=(X,\mathcal{S})$.

\paragraph{New Color Synchronization Gadget.}
First, let us discuss the modification of the Mycielski color synchronization gadget.
Let $M_5$ be Mycielski graph with chromatic number 5 on 23 vertices.
If $V(M_5) = \{u_0, \ldots, u_{22}\}$, then we remove the vertex $u_{16}$; let $M'$ denote the obtained graph and rename the vertices to $v_0, v_1, \dots, v_{21}$ (we shift the numeration of vertices so that the indices are from $0$ to $21$). Then, we are interested in the following quadruples of vertices --- $I := (t_0, t_1, t_2, t_3) = (v_{1}, v_{3}, v_{10}, v_{21})$ and $I' := (t_0', t_1', t_2', t_3') = (v_{19}, v_{17}, v_{11}, v_{5})$; see Figure \ref{pic:mycielski}. They have the following properties:
\begin{itemize}
\item both $I$ and $I'$ are independent sets;
\item $t_it_j'$ is an edge if and only if $i \neq j$;
\item in every 4-coloring of $M'$, both $I$ and $I'$ receive 4 different colors. \\ This holds, because $I = (v_{1}, v_{3}, v_{10}, v_{21})$ was the neighborhood of the deleted vertex in our initial $M_5$. If for some $4$-coloring of $M'$, $I$ spanned less than $4$ colors, one could extend the $4$-coloring of $M'$ to the entire $M_5$, contradiction. The coloring of $I'$ is forced by adjacencies in $M'$. 
\end{itemize}

\begin{figure}
\centering
    \includegraphics[width=0.6\textwidth]{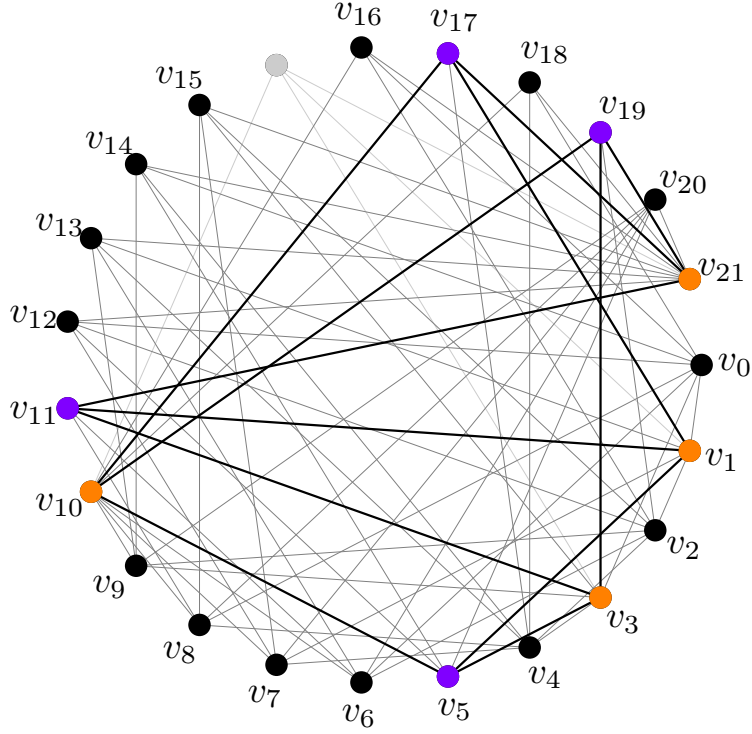}
\caption{Graph $M'$, obtained from Mycielski graph $M_5$ by removing the grayed out vertex. Distinguished colored vertices belong to $I$ and $I'$.}
\label{pic:mycielski}
\end{figure}

Thanks to changes in the color synchronization gadget, in particular, thanks to having two sets of connectors, we do not need to use $W^4$ pendants to enforce the right lists, but we can connect them directly.
Therefore, the clause gadget is defined as follows:

\paragraph{Clause Gadget.}
Define $H$ as a graph on 8 vertices $a_0, b_0, a_1, b_1, a_2, c_0, c_1, c_2$ with edge set $\{a_0b_0, b_0a_1, a_1b_1, b_1a_2, a_0c_0, a_1c_1, a_2c_2\}$. 
We say that $a_0, a_1, a_2$ are of \emph{$a$-type}, $b_0, b_1$ are of \emph{$b$-type}, and $c_0, c_1, c_2$ are of \emph{$c$-type}.
If we remove $b$-type vertices from $H$, it would collapse into three connected components, each consisting of an $a$-type vertex and a $c$-type vertex; we will call such a~component an \emph{$ac$-pair}.

\paragraph{Connecting the Gadgets.}

Now, let $\phi$ be a~formula consisting of clauses $S_0, \ldots, S_m$  (each clause consisting of three variables), and let $x_0, \ldots, x_n$ be all variables that appear in $\phi$.
We will construct a~graph $G = G(\phi)$ as follows: First, we take $M'$ and an independent set of $n$ vertices representing the variables, $x_0, \ldots, x_n$, which will be referred to as vertices of \emph{$x$-type}. 
We add an edge $x_it_2$ and $x_it_3$ for every $1 \leq i \leq n$. 
For each clause $S_j$, we will attach to the construction two copies of $H$, $H_0$ and $H_1$, in the following way. 
If $S_j = (x_{j_0}, x_{j_1}, x_{j_2})$, then we add an edge between $x_{j_h}$ and $c_h \in H_\varepsilon$ for any $h \in \{0,1,2\}$ and $\varepsilon \in \{0,1\}$.
Also, we add an edge $x_iw$ for every $x_i$ and every $b$-type vertex $w$. For every $b$-type vertex $v$, we add edges $vt_0'$ and $vt_1'$, and for every $c$-type vertex $w$, we add edges $wt_2'$ and $wt_3'$. 
Finally, we add edges $a_0t_0, a_0t_2, a_1t_0, a_2t_0, a_2t_3$ for vertices in $H_0$, and $a_0t_1, a_0t_2, a_1t_1, a_2t_1, a_2t_3$ for vertices in $H_1$; since there are six possible neighborhoods in $I$ of an $a$-type vertex, we distinguish six types of an $ac$-pair. Figure \ref{pic:gadgets} shows a~a gadget corresponding to a~single clause.

\begin{figure}
\centering
    \includegraphics[width=0.9\textwidth]{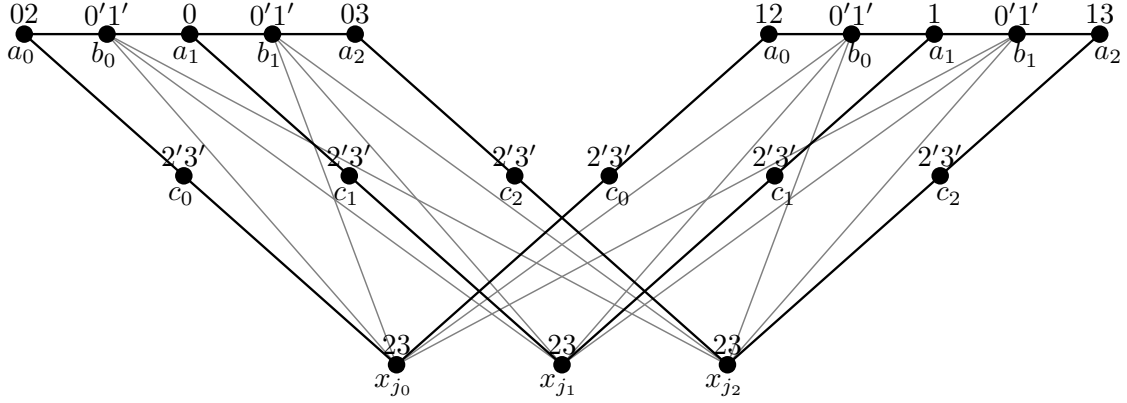}
\caption{A gadget corresponding to a clause $S_j = (x_{j_0}, x_{j_1}, x_{j_2})$. Labels on top of vertices indicate to which vertices from $I$ or $I'$ a given vertex is adjacent to, e.g.~a vertex with a label $0'1'$ is adjacent to $t_0'$ and $t_1'$. Note that $b$-type vertices are adjacent to all $x$-type vertices, including those not present in the gadget.}
\label{pic:gadgets}
\end{figure}

Denote
$$\mathcal{G} := \{G(\phi) : \phi \text{ is an instance of \textsc{MNAE-3-SAT}}\}.$$

\subsection{Correctness and Properties of the Construction}
\begin{proposition}[Triangle-freeness of the new construction in \cref{sec:newConstruction}]\label{prop:triangle-free}
  Any $G \in \mathcal{G}$ \cref{sec:newConstruction} is triangle-free.
\end{proposition}
\begin{proof}
    Fix any $G \in \mathcal{G}$.
    Observe that $M'$ as well as all other gadgets are triangle-free and $M'$ connects to the rest of $G$ only using $I\cup I'$.
    Therefore, it is straightforward to verify that the rest of the construction is also triangle-free as no two adjacent vertices in $G-M'$ connects to the same vertex in $M'$ and no vertex of $G-M'$ connects to $tt'\in E(G)$ where $t\in I$ and $t'\in I'$.
\end{proof}

Since our construction is a relatively simple  modification of the one described in \cref{sec:oldConstruction}, the \NP-completeness proof is similar to that from \cite{HJP15}.
\begin{lemma}[Correctness of the new construction in \cref{sec:newConstruction}]\label{lem:correctness}
  \textsc{$4$-Coloring} is \NP-hard in the graph family $\mathcal{G}$.
\end{lemma}

\begin{proof}
As observed in \cref{sec:newConstruction}, in every 4-coloring of $M'$ the sets $I=\{t_0,t_1,t_2,t_3\}$ and $I'=\{t_0',t_1',t_2',t_3'\}$ are independent and receive four pairwise distinct colors, and moreover each $t_i'$ is forced to have the same color as $t_i$ (this is easy to see since $t_i'$ sees $t_j$ for all $j\neq i$). 
We recall the list assignment $L$ and its subdivision version $L'$ used on $J'_\phi$ in \cref{sec:oldConstruction}. %
We claim that it is straightforward to check that under the fixed coloring of $M'$, every vertex in our new construction has exactly the same set of available colors as the corresponding vertex had in $J'_\phi$ under $L'$. 
Therefore, we can reuse the correctness argument from \cite[Theorem 7]{HJP15}: our two copies $H_0,H_1$ play the role of the two $P_5$ clause components $C_j,C'_j$; $a$-, $b$-, $c$-, and $x$-type vertices have exactly the same admissible colors; and the incidence between variables and the three $ac$-pairs in a clause is also identical.
\end{proof}

Our goal now is to reduce the number of graphs to check for the longest induced path. The proof will be computer-assisted, and 
we aim to show that every induced path must be contained in some~small specific substructure, which can be then examined by a computer program.
We provide the sources of our program with the arXiv submission\footnote{See file \texttt{sage\_script.ipynb}}. %

\begin{lemma} \label{lem:three_x_vertices}
        For any graph $G \in \mathcal{G}$, there are at most three $x$-type vertices in any induced path in G.
    Moreover, if an induced path in $G$ contains three $x$-type vertices, then it has at most $15$ vertices.
\end{lemma}
\begin{proof}

Suppose that some induced path $P$ contains at least three $x$-type vertices. Then, vertices $t_2$ and $t_3$ and vertices of $b$-type cannot appear in $P$ since they are adjacent to all $x$-type vertices, hence $P$ must have the form $P_0 x_1 c^1 P_1 c^2 x_2 c^3 P_2 c^4 x_3 P_3$, where each $c^i$ is of $c$-type. Now it follows that $t_2'$ and $t_3'$ cannot appear in $P$, as they are adjacent to all $c$-type vertices, hence $P_1$ is of the form $a^1 P_1' a^2$ and $P_2$ is of the form $a^3 P_2' a^4$, where $a^i$ are of $a$-type. But now, both $P_1'$ and $P_2'$ must have endpoints in $\{t_0, t_1, t_2, t_3\}$ and no vertex can be used twice --- and since we already excluded the possibility of using $t_2$ or $t_3$, we must have $P_1' = t_0$ and $P_2' = t_1$ (or vice versa).

Observe that neither $P_0$ nor $P_3$ can contain any $a$-type vertex, as it would be adjacent to either $t_0$ or $t_1$, hence both $P_0$ and $P_3$ either are empty or consist of a~single $c$-type vertex only. In particular, any induced path in $G$ can have at most three $x$-type vertices, and every induced path with three $x$-type vertices is of order either 13, 14, or 15. \end{proof}

Suppose $P$ is an induced path in some $G \in \mathcal{G}$ and let $P_1, \ldots, P_k$ be all connected components of $P \cap M'$; we shall call such components the \emph{$M$-components} of $P$. Observe that if $k \geq 2$, then each $M$-component has a~neighbor in $G - M'$, hence contains at least one vertex from $I \cup I'$.
\begin{lemma}\label{lemma:M-components}
    For any $G \in \mathcal{G}$, every induced path $P$ in $G$ contains at most four $M$-components.
\end{lemma}

\begin{proof}
    Assume the contrary, and let $P_1, \ldots, P_5$ be $M$-components of $P$. By the observation above, each $P_i$ contains some vertex from $I \cup I'$. But any five vertices in $I \cup I'$ induce at least one edge and therefore some of the $M$-components would be connected by an edge, which is a~contradiction.  
\end{proof}

We will now analyze all possible induced paths that contain at most two $x$-type vertices. For this reason, we introduce the following constructions:
\begin{itemize}
\item Graphs $G_{0,i}$ (Figure \ref{pic:G_0_i}), for $i = 0, \ldots, 5$, consist of the Mycielski part $M'$ with $i$ copies of $H_0$ and $5-i$ copies of $H_1$ attached to it. %

\begin{figure}
\centering
    \includegraphics[width=0.9\textwidth]{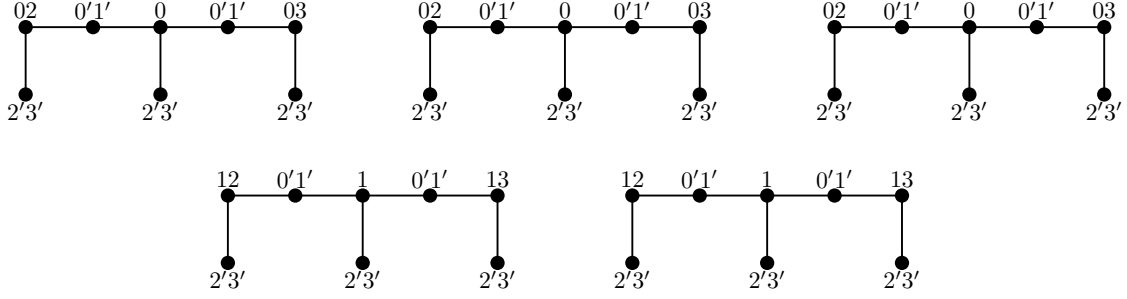}
\caption{Construction $G_{0,3}$ (without the Mycielski part $M'$). Labels at each vertex indicate its neighbors in $I \cup I'$.}
\label{pic:G_0_i}
\end{figure}

\item Graph $G_1$ (Figure \ref{pic:G_1}) is based on the Mycielski part $M'$ and a single $x$-type vertex $v$. We attach to $v$ six $ac$-pairs of all possible types. We also attach two gadgets $H_0$ and $H_1$ in a way that $x_0$ is adjacent only to $b$-type vertices of those gadgets. Finally, we attach to $M'$ additional four $H_0$-gadgets and four $H_1$-gadgets with removed $b$-type vertices. 

\begin{figure}
\centering
    \includegraphics[width=0.9\textwidth]{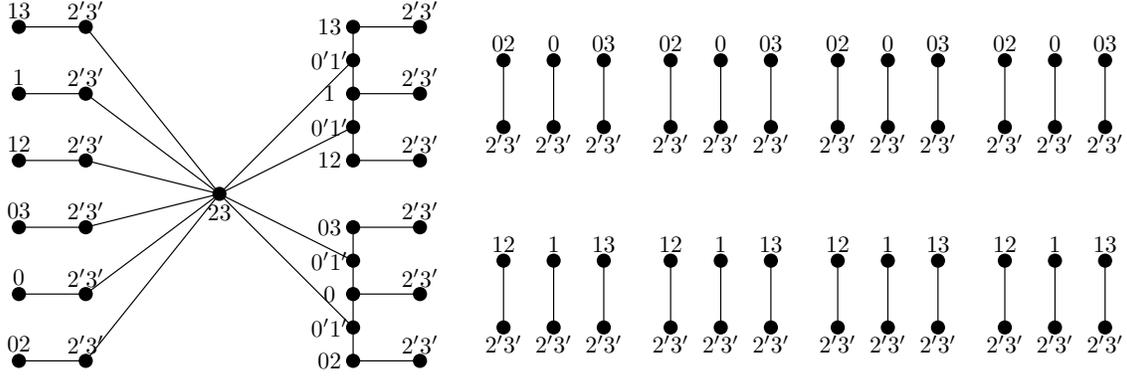}
\caption{Construction $G_1$ (without the Mycielski part $M'$). Labels at each vertex indicate its neighbors in $I \cup I'$.}
\label{pic:G_1}
\end{figure}

\item Graph $G_2$ (Figure \ref{pic:G_2}) is based on the Mycielski part $M'$ and a single $x$-type vertex $v$. We attach two $H_0$-gadgets and two $H_1$-gadgets to $M'$, with $v$ adjacent only to $b$-type vertices. Finally, we attach to $M'$ additional four $H_0$-gadgets and four $H_1$-gadgets with removed $b$-type vertices.

\begin{figure}
\centering
    \includegraphics[width=0.9\textwidth]{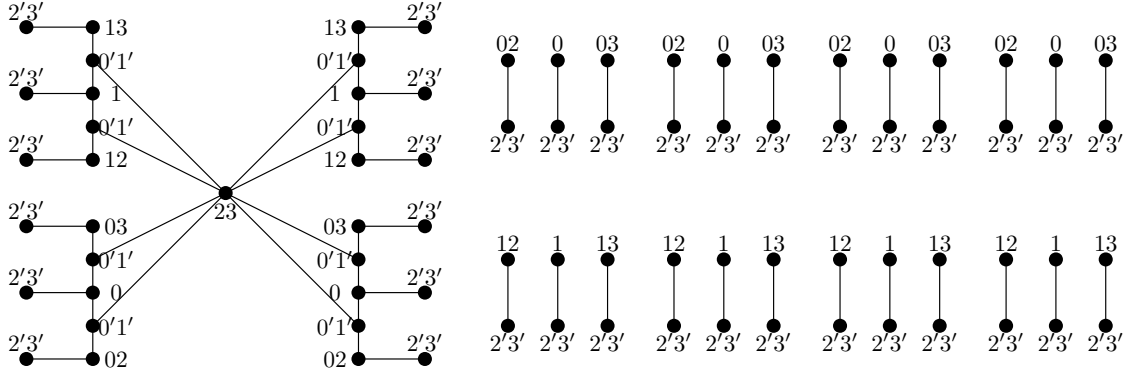}
\caption{Construction $G_2$ (without the Mycielski part $M'$). Labels at each vertex indicate its neighbors in $I \cup I'$.}
\label{pic:G_2}
\end{figure}

\item Graph $G_3$ (Figure \ref{pic:G_3}) is based on the Mycielski part $M'$ and two $x$-type vertices $v, v'$. To each $x$-type vertex we attach two $ac$-pairs of each type. Moreover, we introduce a single $b$-type vertex which is adjacent to both $v$ and $v'$. Finally, we attach to $M'$ four $H_0$-gadgets and four $H_1$-gadgets with removed $b$-type vertices. 

\begin{figure}
\centering
    \includegraphics[width=0.7\textwidth]{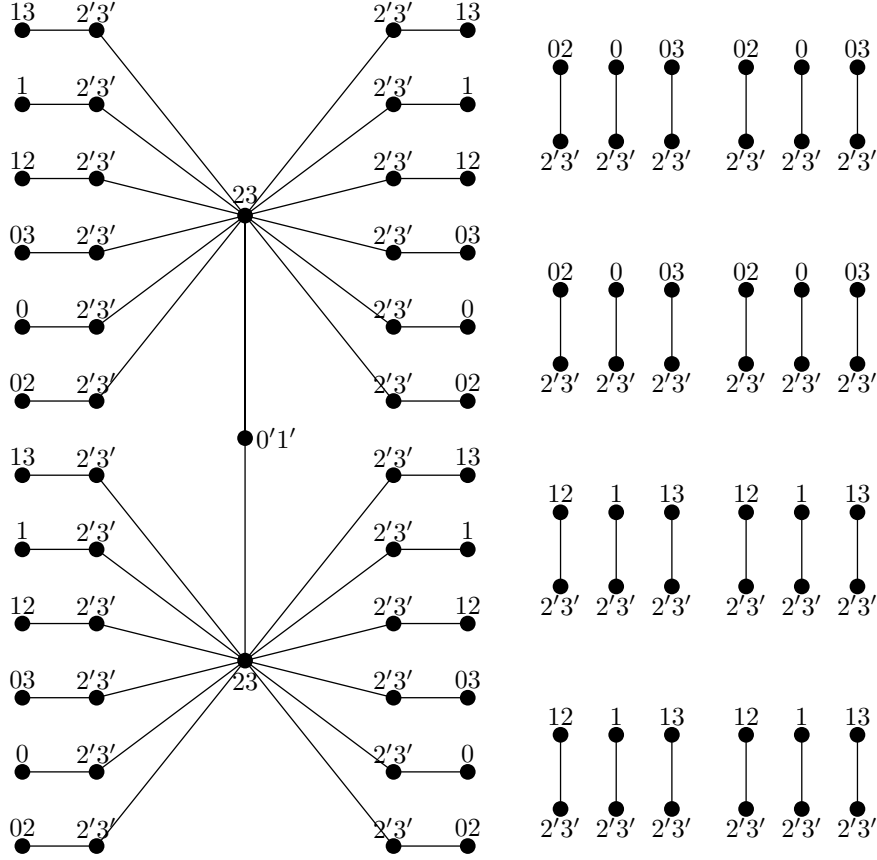}
\caption{Construction $G_3$ (without the Mycielski part $M'$). Labels at each vertex indicate its neighbors in $I \cup I'$.}
\label{pic:G_3}
\end{figure}
\end{itemize}

Let $\Gcore = \{G_{0,i} | i = 0, \ldots, 5\} \cup \{G_1, G_2, G_3\}$ be the set of all constructions defined above. If $P$ is an induced path in some $G \in \mathcal{G}$, then we will say that $P$ can be \emph{realized} in $G' \in \Gcore$ if there exists an injective homomorphism $P \to G'$ which is identity on the Mycielski part and preserves non-edges and types of vertices. 

\begin{lemma}\label{lem:computer_check}
Every $G \in \Gcore$ is $P_{19}$-free.
\end{lemma}

\begin{proof}
  Proof is carried out by computer verification. The sources are provided with the arXiv submission\footnote{See file \texttt{sage\_script.ipynb}}.
\end{proof}

Note that the following path of order 18 can be realized in any $G_{0,i}$: $ababa-t_3-ababac-t_3'-cabab$.

\begin{lemma}\label{lem:path_realization}
    For any $G \in \mathcal{G}$, let $P$ be an induced path in $G$ which contains at most two $x$-type vertices. Then, $P$ can be realized in some $G' \in \Gcore$.
\end{lemma}

\begin{proof}
We distinguish a few cases depending on the number of $x$-type vertices in $G$.

\begin{enumerate}
    \item $P$ contains no $x$-type vertices.

    By Lemma \ref{lemma:M-components}, $P$ may contain at most four $M$-components, and hence it can intersect with at most five gadgets. If $i$ is the number of intersected $H_0$-gadgets, then $P$ can be realized in $G_{0,i}$.

    \item $P$ contains exactly one $x$-type vertex.

    Let $v$ be the only $x$-type vertex in $P$. Observe that every $b$-type vertex in $P$ must be a neighbor of $v$. We distinguish the following subcases:
    \begin{enumerate}
        \item If $v$ has no $b$-type neighbors in $P$, then $P$ contains no $b$-type vertices at all. Moreover, it may intersect with at most six $ac$-pairs --- at most two of them are incident to $v$, and the remaining are separated in $P$ via $M$-components. Therefore, $P$ can be realized in $G_3$.
        \item If $v$ has exactly one $b$-type neighbor, then this $b$-type vertex belongs either to a $H_0$-type or a $H_1$-type gadget. Apart from that, $P$ may intersect with at most five $ac$-pairs --- at most one of them is incident to $v$, and the remaining are separated in $P$ via $M$-components. Therefore, $P$ can be realized in $G_1$.
        \item If $v$ has exactly two $b$-type neighbors, then $v$ is incident to at most two gadgets, and $P$ can intersect with at most four $ac$-pairs disjoint from those gadgets, since it may have at most four $M$-components. Therefore, $P$ can be realized in $G_2$.
    \end{enumerate}
    
    \item $P$ contains two $x$-type vertices. 

    Let $v$ and $v'$ be two $x$-type vertices of $P$ and write $P$ as $P_1vP_2v'P_3$. Observe that neither $P_1$ nor $P_3$ contain any $b$-type vertices. We consider the following subcases.
    \begin{enumerate}
        \item If $P_2$ is a single vertex, then it's either $t_2$, $t_3$, or of $b$-type. In either case, $P_1$ and $P_3$ intersect at most six $ac$-pairs --- at most two of them are adjacent to some $x$-type vertex, and the remaining are joined by $M$-components. Therefore, $P$ can be realized in $G_3$.
        \item If $P_2$ is of the form $ct_2'c'$ or $ct_3c'$ for some $c$-type vertices $c$, $c'$, then $P_1$ and $P_3$ must be empty as they cannot have any $b$-type or $c$-type vertices. In particular, $P$ can be realized in $G_3$.
        \item Otherwise, $P_2$ is of the form $caP_2'a'c'$ for some $ac$-pairs $ac$, $a'c'$, where $P_2$ is either a single vertex from $I$ or has two different endpoints in $I$. Then, $P$ can intersect at most eight $ac$-pairs --- at most four of them are incident to some $x$-type vertex, and the remaining are joined by $M$-components. Therefore, $P$ can be realized in $G_3$.\qedhere
    \end{enumerate}
\end{enumerate}
\end{proof} 

\begin{proof}[Proof of \cref{thm:main}]
  By \cref{lem:correctness}, \textsc{$4$-coloring} is \NP-hard in the class $\mathcal{G}$. By Lemmas \ref{lem:three_x_vertices}, \ref{lem:computer_check}, and \ref{lem:path_realization}, $\mathcal G$ is $P_{19}$-free and \cref{prop:triangle-free} states that the family is triangle-free, which concludes the proof.
\end{proof}

\bigskip
\noindent\textbf{Acknowledgements.}~
T.~Masařík and J.~Masaříková want to thank Andrzej Grzesik for hosting them and for the pleasant and hospitable atmosphere of Jagiellonian University.
The authors also thank Andrzej Grzesik for the valuable discussions that emerged from this work.

\bibliographystyle{alphaurl}
\bibliography{lit}

\end{document}